\documentclass[conference,letterpaper]{IEEEtran}  

\usepackage{amsthm,amssymb,graphicx,multirow,amsmath,color,amsfonts,subfig,balance}
\usepackage{bm}
\usepackage[latin1]{inputenc}
\usepackage[update,prepend]{epstopdf}
\usepackage[noadjust]{cite}
\usepackage{mathtools}
\usepackage[]{algorithm2e}
\usepackage{multirow}
\usepackage{bbm} 
\usepackage{pdfpages}
\usepackage{tabulary}
\usepackage{multirow}
\usepackage{comment}
\usepackage{algorithmicx}
\usepackage{algpseudocode}
\usepackage{etoolbox}
\usepackage{enumerate}
\usepackage[svgnames]{xcolor}
\def\beq{\begin{equation}}
\def\eeq{\end{equation}}
\def\beqa{\begin{eqnarray}}
\def\eeqa{\end{eqnarray}}
\def\beqan{\begin{eqnarray*}}
\def\eeqan{\end{eqnarray*}}
\include{notation_new}
 
\usepackage{tikz}
\usetikzlibrary{chains,arrows,calc,positioning}
\usetikzlibrary{shapes.multipart}
\usetikzlibrary{decorations.pathreplacing}
\usetikzlibrary{calc}

\usepackage[normalem]{ulem} 

\newcommand{\citep}[1]{\cite{#1}}

\def\beq{\begin{equation}}
\def\eeq{\end{equation}}
\def\beqa{\begin{eqnarray}}
\def\eeqa{\end{eqnarray}}
\def\beqan{\begin{eqnarray*}}
\def\eeqan{\end{eqnarray*}}

\def\C{{\mathbb{C}}}

\DeclareMathOperator*{\argmax}{arg\,max}

\newtheorem{proposition}{Proposition}

\newtheorem{theorem}{Theorem}
\newtheorem{lemma}{Lemma}

\theoremstyle{definition}

\def\zhat{\widehat{z}}

\def\Exp{\mathbb{E}}

\newcommand{\abf}{\mathbf{a}}

\newcommand{\gbf}{\mathbf{g}}

\newcommand{\rbf}{\mathbf{r}}

\newcommand{\sbf}{\mathbf{s}}

\newcommand{\ubf}{\mathbf{u}}

\newcommand{\wbf}{\mathbf{w}}

\newcommand{\xbf}{\mathbf{x}}

\newcommand{\ybf}{\mathbf{y}}

\newcommand{\zbf}{\mathbf{z}}

\newcommand{\zbfhat}{\widehat{\mathbf{z}}}

\newcommand{\Fbf}{\mathbf{F}}
\newcommand{\Gbf}{\mathbf{G}}

\newcommand{\Qbf}{\mathbf{Q}}

\newcommand{\Vbf}{\mathbf{V}}

\newcommand{\Zhat}{\widehat{Z}}

\def\alphabf{{\boldsymbol \alpha}}

\def\deltabf{{\boldsymbol \delta}}
\def\nubf{{\boldsymbol \nu}}

\def\nubf{{\boldsymbol \nu}}
\def\xibf{{\boldsymbol \xi}}

\newcommand{\phibf}{{\bm{\phi}}}

\newcommand{\indic}[1]{\mathbbm{1}_{ \{ {#1} \} }}

\newcommand{\herm}{^{\text{\sf H}}}

\newtoggle{conference}
\togglefalse{conference} 

\begin{document}
\bstctlcite{IEEEexample:BSTcontrol}
\title{Capacity Bounds for 
Communication Systems with Quantization and Spectral Constraints}

\author{
\IEEEauthorblockN{
Sourjya Dutta, 
Abbas Khalili, Elza Erkip, Sundeep Rangan} 

\IEEEauthorblockA{Dept. of Electrical and Computer Engineering, \\ Tandon School of Engineering, New York University, Brooklyn, NY, USA}
}

\iftoggle{conference}{}{
\thanks{
The work  supported in part by
NSF grants  1302336,  1564142,  1547332, and 1824434,  NIST, SRC, and the industrial affiliates of NYU WIRELESS.
}

}
\maketitle

\begin{abstract}
Low-resolution digital-to-analog 
and analog-to-digital converters (DACs and ADCs) 
have attracted considerable attention 
in efforts to reduce power consumption in millimeter wave (mmWave) 
and massive MIMO systems.
This paper presents an information-theoretic analysis
with capacity bounds for  classes of linear transceivers with
quantization.  The transmitter modulates
symbols via a unitary transform followed by a DAC
and the receiver employs an ADC
followed by the inverse unitary transform.
If the unitary transform is set to an FFT matrix,
the model naturally captures filtering
and spectral constraints 
which are essential to model in any practical transceiver. In particular, this model allows studying the impact of quantization on out-of-band emission constraints. 
In the limit of a large random unitary transform, it is shown that the effect
of quantization can be precisely described via an additive Gaussian noise model.
This model in turn leads to simple and intuitive expressions for the 
power spectrum of the transmitted signal and a lower bound to the capacity with quantization.  
Comparison with non-quantized capacity and a capacity upper bound that does not make linearity assumptions suggests that while low resolution quantization
has minimal impact on the achievable rate at typical  parameters in 5G systems today, satisfying out-of-band emissions are potentially much more  of a challenge.  
\end{abstract}


\section{Introduction}

All digital communications systems rely on digital-analog and
analog-digital converters (ADCs and DACs).
In recent years, there has been considerable interest
in systems with so-called \emph{low resolution} DACs and ADCs where
the number of bits is very small (typically 3-4 bits in I and Q).
These architectures have attracted particular attention in
the context of energy-efficient approaches
for next-generation millimeter wave (mmWave) and massive MIMO systems
\cite{abbas2017millimeter,zhang2018low,abdelghany2018towards,abbas2019highsnr,abbas2019MT,Abbas2020Thr,barati_initialaccess,singh2009limits,koch2013low,nossek2006capacity,orhan2015low,dutta2019case, mo2015capacity,rini2017general,mezghani2012capacity,Studer2016,jacobsson2017throughput,mollen2016uplink,mo2017hybrid}.
In particular, mmWave systems rely on communication across
wide bandwidths with large numbers of antennas
\cite{rappaport2014millimeter,rangan2014millimeter}.
Power consumption thus becomes a key issue, particularly
in so-called fully digital architectures where signals from all antennas
are digitized for fast beam-tracking, initial access and spatial multiplexing
\cite{barati_initialaccess,abbas2017millimeter,zhang2018low,abdelghany2018towards,dutta2019case}.

At low resolutions, it is critical to evaluate the effect of quantization accurately, and there
is now a large body of work on characterizing the capacity of such systems \cite{abbasISIT2018,mo2015capacity,rini2017general,mezghani2012capacity,singh2009limits,koch2013low,nossek2006capacity,Studer2016,jacobsson2017throughput,mollen2016uplink,mo2017hybrid}. The most common model is to approximate the quantizer in either the
DAC or ADC via an additive Gaussian noise (AGN) model  \cite{gersho2012vector,fletcher2007robust}. There are several works that provide rigorous analysis of the AGN model under variety of assumptions such as the high rate regime or dithered quantization \cite{gersho2012vector,marco2005validity,gray1993dithered,zamir1996LQN,derpich2008quadratic}.
The AGN model has also been used in the analysis of low resolution mmWave systems
\cite{mo2015capacity,rini2017general,mezghani2012capacity,Studer2016,jacobsson2017throughput,mollen2016uplink,mo2017hybrid}. 
In such systems, while the AGN and other 
Gaussian noise predictions match simulations, its use has not been rigorously justified.

This paper presents a simple, but rigorous method,
for analyzing a large class of linear communication
systems. Specifically, we analyze a general
transmitter and receiver with quantization in conjunction with linear modulation
and demodulation as shown in Fig.~\ref{fig:fd_txrx}.
A transmitter encodes data through an unitary transform $\Vbf \herm$ prior
to the DAC.  The DAC is modeled by a function $\Qbf_{\rm tx}(\cdot)$.
The continuous-valued signal $\xbf$ is passed through a memoryless
channel $\Fbf(\cdot)$.  The receiver then uses an ADC $\Qbf_{\rm rx}(\cdot)$ followed by
an inverse transform $\Vbf$ to recover the transmitted symbols.

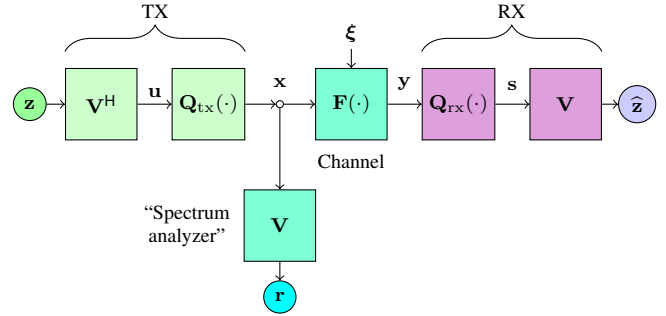
\begin{figure}
\begin{tikzpicture}[scale=0.95,every text node part/.style={align=center}, every node/.append style={transform shape}]
    \footnotesize
    
    \node [draw,circle,fill=green!40, node distance=1cm]
        (z) {$\zbf$};
    \node [draw,right of=z,fill=green!20, minimum size=1cm,
        node distance=1 cm] (VH) {$\Vbf\herm$};
    \node [draw,right of=VH,fill=green!20, minimum size=1cm,
        node distance=1.5 cm] (Q) {$\Qbf_{\rm tx}(\cdot)$};
    \node [draw,circle,right of=Q,inner sep=1] (circ1) {};
    \draw [->] (z) --  (VH);
    \draw [->] (VH) --  node[above] {$\ubf$} (Q);
    \draw [->] (Q) -- (circ1);
    \node [draw=none,above of=circ1,yshift=-0.7cm] ()  {$\xbf$};
    \draw [decorate,decoration={brace,amplitude=10pt,raise=0.7cm}] 
        (VH.west) -- node[above,yshift=1.1cm] {TX} (Q.east) ;

    \node [draw,right of=circ1,fill=Aquamarine, minimum size=1cm,
        node distance=1 cm] (chan) {$\Fbf(\cdot)$};
    \draw [->] (circ1) -- (chan);
    \node [above of=chan]  (xi) {$\xibf$};
    \draw [->] (xi) -- (chan.north);
    \node [below of=chan,yshift=0.2cm]  () {Channel};
    
    \node [draw,below of=circ1,fill=Aquamarine, minimum size=1cm, node distance=1 cm, yshift=-0.7cm] (Vsa) {$\Vbf$};
    \node [draw,below of=Vsa,circle,fill=Aqua, node distance=1cm]
        (r) {$\rbf$};
    \draw [->] (circ1) -- (Vsa);
    \draw [->] (Vsa) -- (r);
    \node [left of=Vsa,xshift=-0.3cm]  () {``Spectrum\\ analyzer"};

    \node [draw,right of=chan,fill=Plum, minimum size=1cm,
        node distance=1.5 cm] (Qrx) {$\Qbf_{\rm rx}(\cdot)$};
    \node [draw,right of=Qrx,fill=Plum, minimum size=1cm,
        node distance=1.5 cm] (V) {$\Vbf$};
    \node [draw,circle,right of=V,fill=blue!20] 
        (zhat) {$\widehat{\zbf}$};
    \draw [->] (chan) --  node[above,yshift=0.1cm] {$\ybf$} (Qrx);
    \draw [->] (Qrx) -- node[above,yshift=0.1cm] {$\sbf$} (V);
    \draw [->] (V) --   (zhat);
    \draw [decorate,decoration={brace,amplitude=10pt,raise=0.7cm}] 
        (Qrx.west) -- node[above,yshift=1.1cm] {RX} (V.east) ; 

\end{tikzpicture}

\caption{System model with transform modulation and demodulation
with quantization at both the transmitter and receiver.
The transform modulation is modeled as a multiplication by
$\Vbf\herm$ prior to quantization at the transmitter, while a  spectrum
analyzer and receiver employ the inverse transform $\Vbf$.}
\label{fig:fd_txrx}

\end{figure}

If $\Vbf$ were an FFT-matrix, then 
the model can be considered as a simplified version of
a frequency-domain filtering.
Also, the spectrum of the transmitted signal can be modeled through
the transform $\rbf = \Vbf\xbf$.
We find an achievable rate for this system and  the power spectral density of the transmitted  signal
as a function of the DAC and ADC functions in a certain large random
limit where $\Vbf \in \C^{N \times N}$ is selected
uniformly among the unitary matrices and $N \rightarrow \infty$. We also find a capacity upper bound for a given transmitted power spectral density considering the DAC and the ADC, but not limiting transmit/receive processing to linear operations.  
Our key results are as follows:
\begin{itemize}
    \item \emph{Rigorous AGN model:} 
    We show that the effect of quantization can be
    precisely modeled as additive, independent Gaussian noise.
    This result makes the AGN analysis of \cite{gersho2012vector} in the setting of Fig.~\ref{fig:fd_txrx}
    rigorous, even in the low rate regime.
    
    \item \emph{Predictions on the rate and power spectrum:}
    The AGN model provides  asymptotically
    exact, simple and intuitive expressions 
    for spectrum of the transmitted 
    signal and a lower bound for the capacity of the quantized channel.
    
    \item \emph{Sampling rate and spectral modeling:}
    Many prior information theoretic analyses of low-resolution communication 
    systems assume that the symbol rate equals the sample rate (see, for example, \cite{singh2009limits,mo2015capacity}).
    However, almost all practical transceivers use a sampling rate higher than the signal bandwidth
    to reduce the filtering requirements in the analog domain.  Oversampling is also
    needed in systems with variable bandwidths where sub-channels are selected digitally
    (see Sec. \ref{sec:num_res}  for an example based on 5G New Radio standard \cite{3GPP38300}).
    Previous works accounting for oversampling consider very specific up-sampling methods
    \cite{krone2010fading}.  In contrast, our methods enable exact calculations of the power spectrum
    and bounds on capacity under general spectral mask constraints.
    
    \item \emph{Implications for fully-digital architectures for 5G New Radio:} 
    Several prior simulation studies have predicted that  
    with 3 -- 4 bits, the loss from quantization in achievable rate
    is minimal for data and control plane operations in most 5G cellular use cases
    \cite{barati_initialaccess,abbas2017millimeter,zhang2018low,abdelghany2018towards,dutta2019case,Studer2016,jacobsson2017throughput,mollen2016uplink,mo2017hybrid}.
    Our analysis provides a rigorous confirmation of this minimal loss in achievable rate.
    However, we also show that simple linear modulation results in a hard limit on the
    degree to which the 
    out-of-band (OOB) noise  can be suppressed.  This OOB noise 
    is, in fact, much more of an issue that the rate loss
    at most practical parameter values in 5G systems
    today, particularly in licensed spectrum deployments
    where adjacent carrier leakage is strictly limited.
    
    \item \emph{Upper bounds on OOB suppression for
    any transmitter}:  The high OOB levels with
    the simple linear modulator raises the
    question if there are any transmitter (possibly
    non-linear) that can provide greater OOB
    suppression.  Interestingly, our capacity upper bound for a given power spectral density closely matches
    the achievable rate by the linear transform 
    transmitter in some regime, but shows possibility for
    greater OOB suppression in other regimes.
\end{itemize}
A full version of this paper can be found in
\cite{dutta2020quantization-arxiv} that includes
all proofs.

\section{System Model}
\label{sec:model}

\subsection{Transceiver with Transform Modulation and Demodulation}

We consider the general transceiver system with quantization and transform modulation 
and demodulation shown 
in Fig.~\ref{fig:fd_txrx}.
The transmitter constructs a vector of $N$
symbols $\zbf=(z_0,\ldots,z_{N-1})$ which are modulated 
as $\ubf=\Vbf\herm \zbf$ where $\Vbf \in \C^{N \times N}$ is some 
unitary matrix.  The transformed values are quantized to result 
in a transmitted
vector $\xbf = \Qbf_{\rm tx}(\ubf)=\Qbf_{\rm tx}(\Vbf\herm \zbf)$, where $\Qbf_{\rm tx}(\cdot)$
models the DAC.  
If $\Vbf$ were an FFT matrix,
we could consider the symbols $\zbf$ as the values of 
the transmitted signal in frequency domain and $\ubf$ the pre-quantized
values in time-domain.  The modulation can thus be regarded as a simplified
version of OFDM (where we ignore the cyclic prefix).
In addition, if we zero-pad the input frequency-domain 
symbols $\zbf$, the transformed vector $\ubf =\Vbf\herm\zbf$
can be seen as an linearly up-sampled version of $\zbf$.  

The transmitted time-domain symbols are passed through a general channel of the form,
\beq \label{eq:chan}
    \ybf = \Fbf(\xbf,\xibf),
\eeq
where $\Fbf(\cdot)$ is some mapping and $\xibf$ is noise independent of the channel input $\xbf$.  Most commonly,
we will be interested in the AWGN case, $\ybf = h\xbf + \xibf$, where $h$
is the channel gain.  The channel \eqref{eq:chan} can also model certain
non-linearites in the RF front-end \cite{abdelghany2018towards}.
The receiver first passes the signal through an ADC $\Qbf_{\rm rx}(\ybf)$
and then performs the inverse transform operation to obtain $\zbfhat = \Vbf\Qbf_{\rm rx}(\ybf)$.

\subsection{Spectrum and Capacity}
\iftoggle{conference}{}{
We are interested in estimating the effect of quantization on
two key quantities: 
the frequency-domain \emph{power spectrum} and the \emph{capacity}.   }

To model the spectrum, let $\rbf = \Vbf\xbf$ which is the transform of the transmitted signal $\xbf$.
The component $|r_k|^2$ can be regarded as the energy of the signal at frequency $k, k=0, \ldots, N-1$.
We assume the frequency is divided into $M$ sub-bands
and let $a_k \in \{1,\ldots,M\}$ be the variable that indicates
which sub-band frequency $k$ belongs to.
We call $\abf=(a_0,\ldots,a_{N-1})$ the \emph{sub-band
selection vector} and let,
\beq \label{eq:delm}
    \delta_m(\abf) := \frac{1}{N} \sum_{k=0}^{N-1} \indic{a_k = m},
\eeq
which represents the fraction of the frequency components in
sub-band $m$.  
\iftoggle{conference}{We will call $\delta_m$ the \emph{bandwidth fraction}
for sub-band $m$.}{} We also define,
\beq \label{eq:sm_nolim}
    \phi_m(\rbf) := \frac{1}{N} \sum_{k=0}^{N-1} \indic{a_k = m}|r_k|^2,
\eeq
which represents the energy per sample in sub-band $m$.

An achievable rate for the system
can be computed by fixing some distribution on $\zbf$
and computing the mutual information $I(\zbf; \zbfhat)$
between the transmitted vectors $\zbf$ and received 
frequency-domain vectors, $\zbfhat$.
For the input distribution, we will use an independent complex Gaussian in
each frequency.  Specifically, we will assume the components
$z_k$ are independent with,
\beq \label{eq:zmix}
    z_k \sim {\mathcal CN}(0,P_m)  \mbox{ when }
    a_k = m,
\eeq
where $P_m$ is the symbol energy on any component 
in sub-band $m$.  The average per symbol energy is,
\beq \label{eq:Pbar}
    \overline{P} = \frac{1}{N} \Exp\|\zbf\|^2 = 
    \frac{1}{N} \Exp\|\ubf\|^2 = \sum_{m=1}^M \delta_mP_m,
\eeq
where $\delta_m$ are the bandwidth fractions \eqref{eq:delm}.

\section{Achievable Spectral Energy and Rate}  \label{sec:lsl}

\subsection{Large System Limit} 
To make the analysis tractable, we consider a certain large system limit
of random instances of the system indexed by the dimension $N$ with $N \rightarrow \infty$.
For each $N$, instead of considering the deterministic
FFT matrix $\Vbf$, we suppose that
 $\Vbf=\Vbf(N)$ is a random unitary matrix 
that is uniformly distributed on the $N \times N$ unitary matrices i.e., Haar distributed.
The sub-band selection vectors 
$\abf=\abf(N)$  are assumed to be a deterministic
sequence satisfying,
\beq \label{eq:delm_limit}
    \lim_{N \rightarrow \infty} \frac{1}{N} |\{ a_k(N) = m \}| = \delta_m.
\eeq
The condition \eqref{eq:delm_limit} imposes that asymptotically 
a fraction $\delta_m$ of the components are in sub-band $m$.  

For the DAC function, $\Qbf_{\rm tx}(\ubf)$, we require that it is Lipschitz continuous
and \emph{componentwise separable} (or, equivalently memoryless
operation) meaning that
\beq \label{eq:Qcomp}
    \xbf = \Qbf_{\rm tx}(\ubf) \Longleftrightarrow x_n = Q_{\rm tx}(u_n),
\eeq
for some scalar-input, scalar-output function $Q_{\rm tx}(\cdot)$.  The componentwise
function $Q_{\rm tx}(\cdot)$ does not change with $N$.  Similarly, we assume
that the channel $\Fbf(\cdot)$ and receiver ADC function act componentwise
with Lipschitz functions $F(\cdot)$ and $Q_{\rm rx}(\cdot)$.  
This corresponds to a memoryless channel. 
Typical quantizers
are not Lipschitz continuous,  but they can be approximated arbitrarily closely
by a Lipschitz function.  We will validate through simulations in Sec. \ref{sec:num_res} that our predictions hold true even
for standard discontinuous quantizers.

\subsection{Achievable Spectral Energy Distributions}
We first compute the asymptotic power spectral distribution of the transmitted symbols $\xbf$.
We define:
\beq \label{eq:alpha_tau_tx}
    \alpha_{\rm tx} := \frac{1}{\overline{P}}\Exp\left[ Q_{\rm tx}^*(U)U \right], 
    \quad 
    \tau_{\rm tx} := \frac{1}{\overline{P}}\Exp|Q_{\rm tx}(U)-\alpha_{\rm tx}U|^2,
\eeq
where $\overline{P}$ is the average per symbol energy in $\zbf$ in 
\eqref{eq:Pbar}, $Q_{\rm tx}^*(U)$ is the complex conjugate of $Q_{\rm tx}(U)$
and the expectation in \eqref{eq:alpha_tau_tx}
is over $U \sim C{\mathcal N}(0,\overline{P})$.

\begin{theorem}  
\label{thm:spec_fd}
Under the above assumptions, let $\rbf=\Vbf\xbf$
be the frequency-domain representation of the transmitted signal $\xbf$.
Then the energy in each sub-band converges almost surely to,
\begin{align} 
    s_m &:= \lim_{N \rightarrow \infty} \frac{1}{N}\sum_{k=0}^{N-1}
        |r_k|^2\indic{a_k=m} \nonumber \\
        &= \delta_m\left[ |\alpha_{\rm tx}|^2P_m + \tau_{\rm tx} \overline{P} \right].
        \label{eq:sm_fd}
\end{align}
In particular, the total energy per symbol converges almost
surely as,
\begin{align} 
    s_{\rm tot} &:= 
    \lim_{N \rightarrow \infty} \frac{1}{N}\|\xbf\|^2 =
         \left( |\alpha_{\rm tx}|^2+ \tau_{\rm tx}\right)\overline{P}.
        \label{eq:stot_fd}
\end{align}
\end{theorem}

\iftoggle{conference}{}{
\begin{proof}
See Appendix~\ref{sec:spec_proof}.
\end{proof}
}

The proof of Theorem~\ref{thm:spec_fd}
\iftoggle{conference}{in the full paper \cite{dutta2020quantization-arxiv}}{} shows, in fact, that the
frequency-domain representation of the transmitted symbols can be
written as 
\beq \label{eq:awgn_tx}
    \rbf = \Vbf\xbf = \alpha_{\rm tx}\zbf + \wbf_{\rm tx},
\eeq
where $\wbf_{\rm tx}$ has components that are asymptotically independent of $\zbf$ and ``Gaussian-like"
with distribution ${\mathcal CN}(0,\tau_{\rm tx}\overline{P})$.  The vector $\wbf_{\rm tx}$
can be thought as the transmitter quantization noise.
The precise sense in which 
$\wbf_{\rm tx}$ is Gaussian-like is given is somewhat technical and given in the \iftoggle{conference}{full paper \cite{dutta2020quantization-arxiv}}{Appendix}.
What is relevant is that the effect of quantizing and returning to 
frequency domain has the effect of scaling the signal $\zbf$ and adding Gaussian noise.
This makes precise the AGN model in \cite{gersho2012vector,fletcher2007robust}
used in several prior analyzes of low-resolution digital architectures
\cite{barati_initialaccess,dutta2019case}.

From Theorem~\ref{thm:spec_fd}, we see that the fraction of power in sub-band $m$ is,
\beq \label{eq:nu_fd}
    \nu_m := \frac{s_m}{s_{\rm tot}} = 
        \frac{\delta_m(|\alpha_{\rm tx}|^2 P_m/\overline{P} + \tau_{\rm tx})}
        {|\alpha_{\rm tx}|^2 + \tau_{\rm tx}}.
\eeq
For a given DAC function $Q_{\rm tx}(\cdot)$ and input power
level $\overline{P}$, it is shown in  
\iftoggle{conference}{the full paper \cite{dutta2020quantization-arxiv}}{Appendix~\ref{sec:linear_feas}}
that there exists power levels $P_m$ resulting in an 
energy fraction vector $\nubf = (\nu_1,\ldots,\nu_M)$ 
if and only if $\nu_m \geq 0$, $\sum_m \nu_m = 1$
and
\beq \label{eq:nu_fd_feas}
    \nu_m \geq 
        \frac{\delta_m\tau_{\rm tx}}
        {|\alpha_{\rm tx}|^2  + \tau_{\rm tx}}.
\eeq
We will call the set of $\nubf$ satisfying 
these constraints  \emph{linear feasible set}. 
\iftoggle{conference}{}{
Note that \eqref{eq:nu_fd_feas} shows 
there is a lower bound on the energy in any sub-band.  
This arises, intuitively, from the fact that the
quantization noise is white and places energy across the
spectrum.  We will see below that this results in high
OOB emissions settings where
the sampling rate is higher than the signal bandwidth.  
}

\subsection{Achievable Rate}
We next compute the asymptotic achievable rate given by the per symbol mutual information between the transmitted symbols $\zbf$ and received symbols $\zbfhat$:
\begin{align} \label{eq:lin_rate_def}
    R_{\rm lin} := \liminf_{N \rightarrow \infty} \frac{1}{N} I(\zbf;\zbfhat),
\end{align}
We will call this the \emph{linear rate}, since it would be the rate achievable by
the linear transmitter and receiver in Fig.~\ref{fig:fd_txrx}.
Assuming  the components of the noise $\xi_n$ are i.i.d.\ with some distribution
$\xi_n \sim \Xi$ with $\Exp|\Xi|^2 < \infty$, similar to \eqref{eq:alpha_tau_tx}, we define
\beq \label{eq:alpha_tau_rx}
    \alpha_{\rm rx} := \frac{1}{\overline{P}}\Exp\left[ S^*U \right], 
    \quad 
    \tau_{\rm rx} :=  \frac{1}{\overline{P}}\Exp|S-\alpha_{\rm rx}U|^2,
\eeq
where $S$ is the complex random variable, 
\beq 
    S = Q_{\rm rx}\left(F(Q_{\rm tx}(U), \Xi)\right), \quad U \sim C{\mathcal N}(0,\overline{P}),
\eeq
$S^*$ is the complex conjugate of $S$, and  $U$ is independent of $\Xi$.

\begin{theorem} \label{thm:rate_lin}
Under the above assumptions, the linear rate is almost surely bounded below by,
\beq \label{eq:rate_lin}
    R_{\rm lin} \geq 
    \sum_{m=1}^M \delta_m \log\left(1 + \frac{|\alpha_{\rm rx}|^2P_m}{\tau_{\rm rx} \overline{P}} \right).
\eeq
\end{theorem}
\iftoggle{conference}{}{
\begin{proof}
See Appendix~\ref{sec:rate_lin_proof}.
\end{proof}
}

\iftoggle{conference}{It is shown in the full paper
\cite{dutta2020quantization-arxiv} that this bounds
also arise from a simple AGN model of the
transceiver.}{
The rate has a simple interpretation. It is shown in
Appendix~\ref{sec:rate_lin_proof}
that the received symbols are given by,
\beq \label{eq:awgn_rx}
    \zbfhat = \alpha_{\rm rx}\zbf + \wbf_{\rm rx},
\eeq
where $\wbf_{\rm rx}$ is asymptotically independent of $\zbf$ and ``Gaussian-like" with components $C{\mathcal N}(0,\tau_{\rm rx}\overline{P})$
and can be seen as representing the combined effect of the noise in the channel
as well as the DAC and ADC quantization noise.
Similar to Theorem \ref{thm:spec_fd}, the precise sense in which $\wbf_{\rm rx}$ 
is asymptotically Gaussian is given in the proof.
Since $\zbf$ has power $P_m$
in sub-band $m$, the rate lower bound \eqref{eq:rate_lin} is simply the Gaussian capacity 
under the AWGN model \eqref{eq:awgn_rx}.
}
Note that the presented lower bound is achieved using Gaussian inputs. However, as we will show in Sec. \ref{sec:quan_UB}, using Gaussian inputs is not optimal since it does not achieve the maximum high SNR rate. Finding the optimal input distribution is left for future work.

\subsection{Achievable Rate in an AWGN Channel} \label{sec:rate_awgn}
It is useful to consider the special case when we have an additive white Gaussian noise (AWGN) channel
modeled with the function $F(X,\Xi) = X+\Xi$ and $\Xi \sim C{\mathcal N}(0,\sigma^2)$.  Also, 
to make the calculations simple, suppose we assume there is no quantization at the receiver
so that $Q_{\rm rx}(y_n) = y_n$.  Substituting these distributions into \eqref{eq:alpha_tau_rx},
and using the expressions in \eqref{eq:alpha_tau_tx}, we can show that
\beq \label{eq:alpha_tau_awgn}
    \alpha_{\rm rx} = \alpha_{\rm tx}, \quad
    \tau_{\rm rx}  = \tau_{\rm tx} + \frac{\sigma^2}{\overline{P}}.
\eeq
Substituting these values into \eqref{eq:rate_lin}, we obtain,
\beq \label{eq:rate_lin_awgn}
    R_{\rm lin} \geq 
    \sum_{m=1}^M \delta_m \log\left(1 + \frac{|\alpha_{\rm tx}|^2P_m}{\tau_{\rm tx} \overline{P} + \sigma^2} \right).
\eeq
Hence we get the AWGN capacity with a loss from the DAC quantization noise.

\subsection{Achievable Rate When There is No Noise} \label{sec:lin_noise_free}

\iftoggle{conference}{}{
We now consider  the noise-free case.
}
\begin{theorem}  \label{thm:rate_no_noise}
In an AWGN channel if $\sigma^2=0$, then the 
rate bound in \eqref{eq:rate_lin_awgn} 
is given by,
\beq \label{eq:rate_no_noise}
    R_{\rm lin} \geq \log\left( 1 + \frac{|\alpha_{\rm tx}|^2}{\tau_{\rm tx}} \right) - D(\deltabf \| \nubf),
\eeq
for any set of power distributions $\nu_m$ is given by
\eqref{eq:nu_fd}.
\end{theorem}
\iftoggle{conference}{}{
\begin{proof}
See Appendix~\ref{sec:rate_no_noise_proof}.
\end{proof}
}

Even with no noise, the rate is finite
since linear processing results in Gaussian-like
quantization noise.  Also, the linear rate in \eqref{eq:rate_no_noise} is only achievable for feasible
power allocations \eqref{eq:nu_fd_feas}.

\iftoggle{conference}{}{
The rate bound \eqref{eq:rate_no_noise} has an 
interesting interpretation.
The first term on the right hand side of
\eqref{eq:rate_no_noise}, 
$\log(1+ |\alpha_{\rm tx}|^2/\tau_{\rm tx})$, 
is the rate in \eqref{eq:rate_lin_awgn}
if the energies in the sub-bands were allocated evenly,  $P_m = \overline{P}$ for all $m$.
Also, observe that from \eqref{eq:nu_fd}, when $P_m=\overline{P}$, $\nu_m = \delta_m$.
So the case of $\nubf=\deltabf$ corresponds to the equal power allocation case.
The second term, $D(\deltabf\|\nubf)$, in the right hand side of \eqref{eq:rate_no_noise} 
is a measure of the loss as a result of non-uniformly allocating the power.
In particular, if one attempts to reduce the power in some sub-band (e.g.\ it is an adjacent
carrier), there will be a linear modulation rate penalty.
}

\section{Quantized Capacity Upper Bound}
\label{sec:quan_UB}
The results above show  that a linear transceiver in conjunction with quantization
limits system performance in two key ways:  (a) there is a limit
\eqref{eq:nu_fd_feas} to which OOB emissions can be suppressed; and (b)
even in the regimes in which a desired spectral mask is feasible, there is a rate penalty due to quantization noise.
These shortcomings 
raise the question of whether there are transceivers (possibly non-linear) that can achieve
better rate under quantization constraints.
To understand this, consider again transmitting on $N$ complex symbols,
$\xbf = (x_0,\ldots,x_{N-1})$.  Model the DAC constraint
as a constraint, $x_n \in A$ where $A \subset \C$ 
are the possible values of the (complex) DAC. 
We will write this constraint as,
\beq \label{eq:dac_con}
    \xbf \in A^N := \left\{ \xbf ~|~
    x_n \in A \right\},
\eeq
To impose the spectral mask constraints, let 
$\sbf=(s_1,\ldots,s_M)$ be a vector of target energies in each sub-band.
Recall that $\phi_m(\Vbf\xbf)$ in \eqref{eq:sm_nolim} is the energy
in a sub-band for a transmitted vector $\xbf$.
Thus, the set 
\beq
    G_N(\Vbf,\epsilon) 
     := \left\{
        \xbf \in A^N ~|~ \phi_m(\Vbf\xbf) \in 
            [s_m-\epsilon,s_m]~\forall m \right\},          \label{eq:Fneps} 
\eeq
represents the set of vectors $\xbf$ satisfying the DAC constraint
and the sub-band energy constraints within some tolerance $\epsilon > 0$.
If we restrict the modulation to vectors in the set 
$G_N(\Vbf,\epsilon)$, then the maximum rate
any modulation method can obtain is,
\beq \label{eq:rate}
    R_N(\Vbf,\epsilon) := \frac{1}{N}
        \log\left| G_N(\Vbf,\epsilon) \right|,
\eeq
where  $|G_N(\Vbf,\epsilon)|$ is the cardinality of $G_N(\Vbf,\epsilon)$.

As before, assume $\Vbf \in \C^{N \times N}$ is Haar-distributed on the unitary matrices.
Since $\Vbf$ is random, the rate $R_N(\Vbf,\epsilon)$ in \eqref{eq:rate} 
is also random.  We can use Jensen's inequality to upper bound the
expected rate,
\begin{align*}
    \Exp R_N(\Vbf,\epsilon) &= \frac{1}{N} 
        \Exp \log |G_N(\Vbf,\epsilon)| 
        \leq
        \frac{1}{N} 
         \log \Exp|G_N(\Vbf,\epsilon)|.
\end{align*}
Here, the expectation is over $\Vbf$.
We will be interested in the asymptotic value of this upper bound,
\beq \label{eq:rate_upper}
    \overline{R}:= \lim_{\epsilon \rightarrow 0} 
    \lim_{N \rightarrow \infty}
        \frac{1}{N} 
         \log \Exp|G_N(\Vbf,\epsilon)|.
\eeq
In this definition, we take the limit $\epsilon \rightarrow 0$
to ensure that the modulator asymptotically matches
the target sub-band energy levels exactly.  Note
that the order of the limits over $N$ and $\epsilon$
is important.

\begin{theorem} \label{thm:rate_upper}
Let $\sbf=(s_1,\ldots,s_M)$ be a set of target 
sub-band energy levels. We define $s_{\rm tot}$ as the total energy,
and $\nubf=(\nu_1,\ldots,\nu_M)$ as the vector
of energy distributions
\beq \label{eq:s_nu}
    s_{\rm tot} := \sum_{m=1}^M s_m, \quad
    \nu_m = \frac{s_m}{s_{\rm tot}}.
\eeq
Then, under the above
assumptions, the
asymptotic rate upper bound in \eqref{eq:rate_upper} is given by,
\beq \label{eq:rate_limit}
        \overline{R} = H_{\rm max}(s_{\rm tot}) - D(\deltabf \| \nubf).
\eeq
Here $H_{\rm max}(s)$ is given by
\beq \label{eq:Hmax_def}
   H_{\rm max}(s) = \max_V H(V) \mbox{ s.t. } \Exp|V|^2 = s,
\eeq
where the maximization is over all discrete random variables $V$ on the set $A$ with
second moment $\Exp|V|^2 = s$.  
\end{theorem}

\iftoggle{conference}{We see that the upper bound in 
Theorem~\ref{thm:rate_upper} and the achievable noise-free rate
in Theorem~\ref{thm:rate_no_noise} have a similar form, but with
a constant gap and the fact that achievable rate with
linear transceiver are limited to the feasible region
\eqref{eq:nu_fd_feas}.
}
{
The rate upper bound in \eqref{eq:rate_limit} has a natural interpretation.
The term $H_{\rm max}(s_{\rm tot})$ is the maximum entropy we could obtain if we are restricted
to the DAC constellation $A$ and need to achieve a certain total power $s_{\rm tot}$.
If we select the symbols of $x_n$ from the distribution that achieves this entropy,
we would obtain an output spectrum that is flat.  If we need to have a non-uniform power spectrum,
we pay an additional penalty $D(\deltabf\|\nubf)$.
The term $D(\deltabf\|\nubf)$ is precisely
the power distribution loss we saw in the linear rate lower bound \eqref{eq:rate_no_noise}. 
Note that as in Theorem~\ref{thm:rate_no_noise}, Theorem~\ref{thm:rate_upper} applies to the no-noise case. Comparing the rate lower and upper bounds in these theorems, 
we see that there is a gap,
\beq \label{eq:gap}
    \overline{R} - R_{\rm lin} \leq H_{\rm max}(s_{\rm tot})
        - \log\left( 1 + \frac{|\alpha_{\rm tx}|^2}{\tau_{\rm tx}} \right).
\eeq
We will see in the simulations below that for most practical values, this gap is less than one bit.
}

\section{Numerical Results}
\label{sec:num_res}

\begin{figure}
    \centering
    \includegraphics[width=0.36\textwidth]{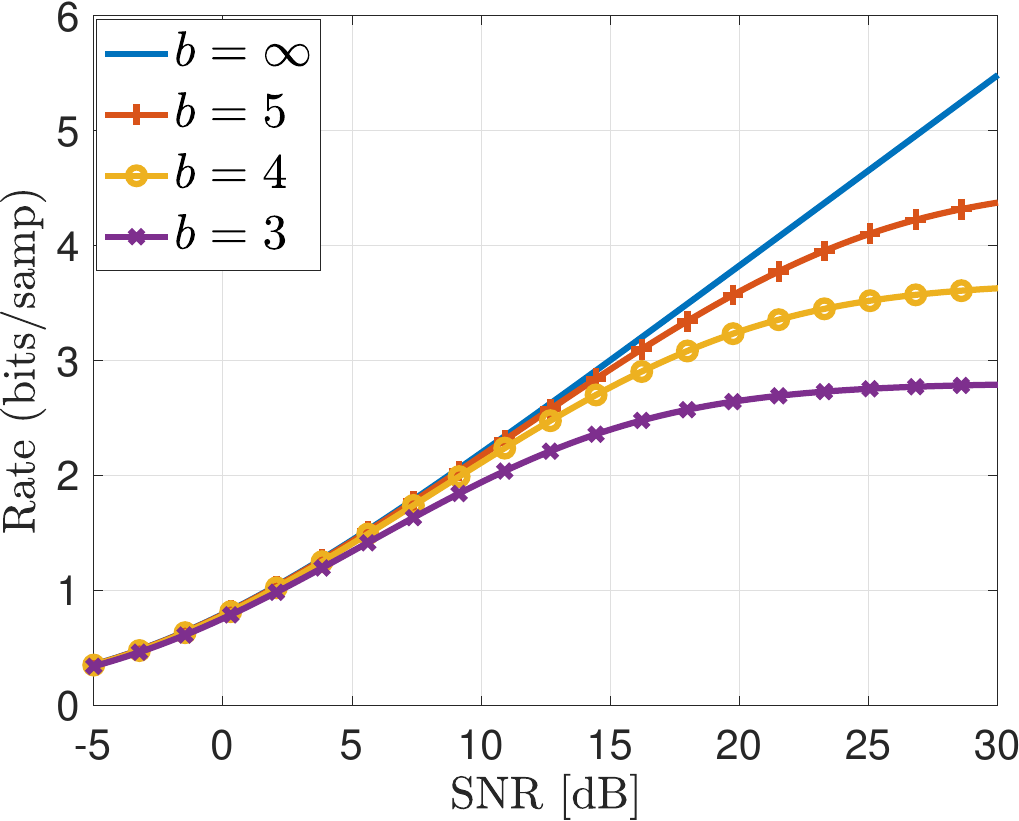}
    \caption{Achievable rate of a system where all the transmit power is allocated to one of two sub-band for different number of DAC bits.}
    \label{fig:rate_snr}
\end{figure}

\begin{figure}
    \centering
    \includegraphics[width=0.36\textwidth]{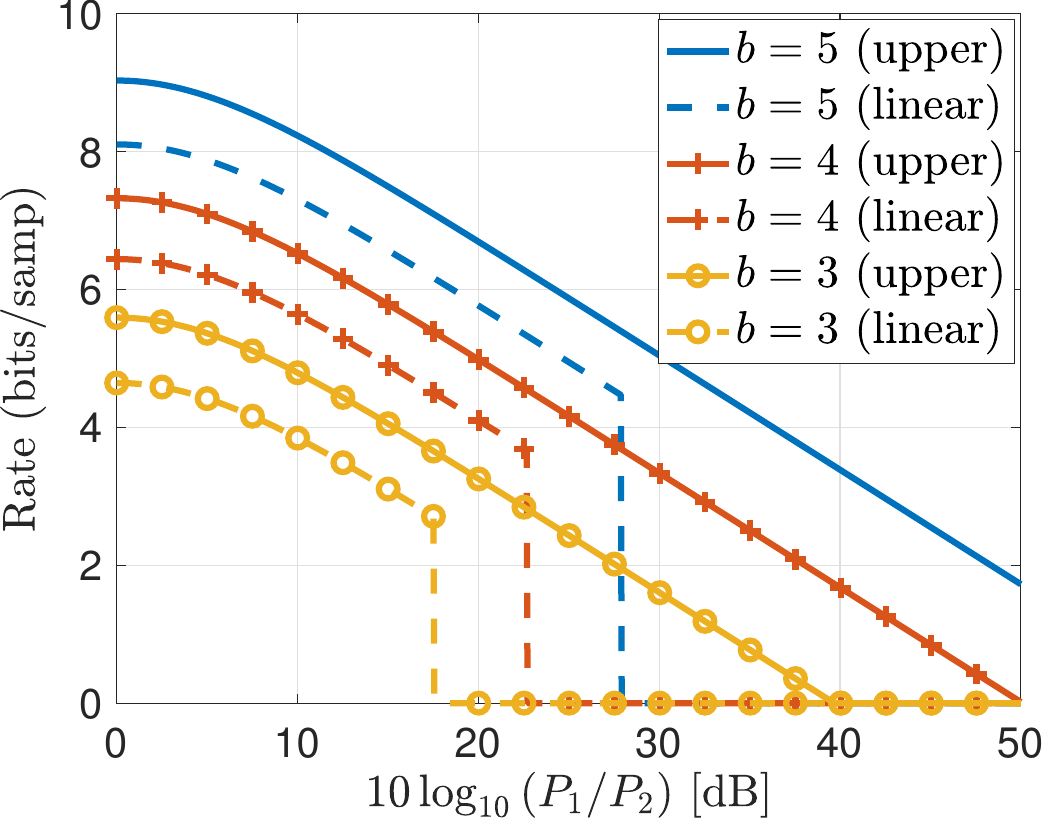}
    \caption{Rate versus adjacent channel leakage in a two sub-band system.  The solid lines show the upper bounds on the achievable rate (Theorem~\ref{thm:rate_upper}) 
    and the dashed lines show the achievable rate predicted by the linear AGN model (Theorem~\ref{thm:rate_no_noise}).}
    \label{fig:rate_aclr}
\end{figure}



 To illustrate
 the results, consider a system where the transmission bandwidth is divided into two equal sub-bands of normalized widths $\delta_1 = \delta_2 = 0.5$. The base-band signal $\ubf$ is designed such that all its energy is concentrated over the first sub-band (representing an in-band signal). Any leakage into sub-band $2$ (representing an adjacent band) is undesirable. Most wireless standard specify a minimum ratio of the in-band to the adjacent band power which defines the \emph{spectrum mask}. The transmitter is equipped with a $b$-bit DAC. The finite resolution of the DAC introduces quantization noise both in-band and in the adjacent carrier.

The effect of the quantization noise on the in-band signal is shown in Fig. \ref{fig:rate_snr}. The achievable rate over 
an AWGN channel for different SNRs and DAC resolutions ($b$) is computed using \eqref{eq:rate_lin_awgn}
assuming a scalar uniform quantizer in both real and imaginary components (I and Q).  We observe that as the resolution of the DAC increases the achievable rate of the system becomes closer to the ideal AWGN capacity (i.e., $b = \infty$). Note that the high SNR achievable rate approaches $b$ bits per sample instead of $2b$ ($b$ bits from in-phase and $b$ bits from quadrature components) since half of the bandwidth is used due to spectral mask constraints. More interestingly, we see that in the low SNR regime there is very little or no loss in rate due to low resolution quantizers. Practical mmWave systems generally operate at the low SNR range \cite{dutta2019case}, particularly when SNR is achieved with beamforming.  The results thus confirm that the rate loss will be negligible in typical low-SNR cellular settings as observed in extensive 
simulations mentioned earlier
\cite{barati_initialaccess,mo2015capacity,rini2017general,mezghani2012capacity,singh2009limits,koch2013low,nossek2006capacity,Studer2016,jacobsson2017throughput,mollen2016uplink,mo2017hybrid}.

On the other hand, a more serious issue is 
the spectral mask constraint. 
Fig.~\ref{fig:rate_aclr} plots the no-noise achievable
rate from \eqref{eq:rate_no_noise} as a function of the signal to adjacent
power, $P_1/P_2$, sometimes called the 
adjacent carrier leakage ratio (ACLR).
We see that, with linear modulation, the
maximum ACLR with non-zero rate is strictly limited.
Fig.~\ref{fig:rate_aclr} also plots the theoretical
maximum rate vs.\ ACLR from Theorem~\ref{thm:rate_upper}.  In the feasible
regime, the linear rate is within one bit of this
upper bound.  But, the upper bound at least 
permits higher ACLRs suggesting that 
more advanced transmitters may be able to 
suppress OOB emissions further.

\paragraph*{Practical low resolution 5G Systems}
Our theory applies to a theoretical 
random transform model.  
\iftoggle{conference}{We illustrate the
model's predictive capabilities in a simulation of 
5G New Radio (NR) \cite{3GPP38300} configured to transmit a
200~MHz channel with a sampling rate of 983 Ms/s, a common
parameter setting in a multi-carrier deployment. 
Fig.~\ref{fig:aclrVsRes} shows the measured ACLR and compares the simulated system with linear AGN model in Theorem~\ref{thm:spec_fd}. We see that the predictions are accurate with $\approx 1$~dB.  See details 
in the full paper \cite{dutta2020quantization-arxiv}.}{
We illustrate the
model's predictive capabilities in a practical 
transceiver shown in 
in Fig.~\ref{fig:linear_mod}.
We consider typical 
for multi-carrier operation 
in the 5G New Radio (NR) standard \cite{3GPP38300} 
using common parameters for 28~GHz \cite{3GPP38104}.
To accommodate variable bandwidths,
the DAC is typically run at a maximum sample
rate.  In this case, we assume an NR standard rate of
$f_{\rm samp} = 2\times 491.5 = 983$~MHz.
A mobile may be allocated a smaller bandwidth, 
say 200 MHz, which would be produced in the NR standard via an OFDM signal at $2 \times 122.6$~MHz.
The modulated baseband signal would be then digitally upconverted to the sampling rate of $f_{\rm samp}$~MHz and digitally filtered to reject spectral images. This interpolated signal is passed through a $b$-bit DAC.

Fig.~\ref{fig:dac_psd} shows the  
output power spectral 
density (PSD) under various numbers of bits in the DAC.
We see that the low DAC resolution creates quantization
noise across the entire bandwidth. The level of that noise increases as the DAC resolution ($b$) is lowered. Moreover, the OOB noise has a flat spectrum (with some decay due to the zero order hold circuit) and justifies the Gaussian model in (\ref{eq:awgn_tx}).

Next, Fig.~\ref{fig:aclrVsRes} shows the ratio of the in-band power $(P_1)$ to the power ``leaked'' into the adjacent band ($P_2$) and compares the simulated system with linear AGN model in Theorem~\ref{thm:spec_fd}. We see that the AGN model is within $\approx 1$~dB of the simulated adjacent channel leakage ratio. The error comes from the fact that the practical simulation models a zero order hold circuit which attenuates some of the OOB noise. Further, the NR OFDM specifications includes a guard band ($\approx 10$~MHz) which is not included in the theoretical calculations. 
Otherwise, we see that the theoretical model
provides an excellent prediction of the spectrum
in a practical low-resolution transmitter. 
}


\iftoggle{conference}{}{
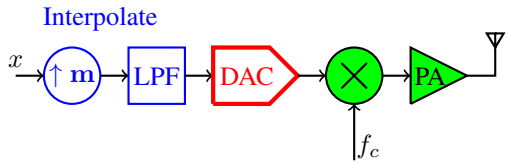
\begin{figure}
    \centering
    \begin{tikzpicture} [scale=0.75]
\draw[thick,->] (-0.5,0.5) -- (0,0.5) node [align=center] at (-0.5,0.75) {$x$}; 
\draw[thick,blue] (0.5,0.5) circle (0.5) node [align=center] at (0.5,0.5) { $\mathbf{\uparrow m}$} node [align=center] at (1,1.5) {Interpolate} ;
\draw[thick,->] (1,0.5) -- (1.5,0.5); 
\draw[thick,blue] (1.5,0) rectangle (2.5,1) node [align=center] at (2,0.5) {LPF};
\draw[thick, ->] (2.5,0.5) -- (3,0.5); 
\draw[ultra thick, red] (3,0) -- (3,1) -- (4,1) -- (4.5,0.5) -- (4,0) -- (3,0) node [align=center] at (3.6,0.5) {DAC};
\draw[thick, ->] (4.5,0.5) -- (5,0.5);
\draw[thick, fill=green] (5.5,0.5) circle (0.5) node [align=center] at (5.5,0.5) {\huge $\times$};
\draw[thick,->] (5.5,-1) -- (5.5,-0) node [align=center] at (5.75,-0.75) {$f_c$};
\draw[thick, ->] (6,0.5) -- (6.5,0.5);
\draw[thick, fill=green] (6.5,0) -- (6.5,1) -- (7.5,0.5) -- (6.5,0) node [align=center] at (6.85,0.5) {PA};
\draw[thick] (7.5,0.5) -- (8.0,0.5) -- (8, 1.25) -- (7.85,1.25) -- (8,1) -- (8.15,1.25) -- (8,1.25);
\end{tikzpicture}
    \caption{Simplified diagram showing standard linear upconversion
    and transmission.}
    \label{fig:linear_mod}
\end{figure}

\begin{figure}
    \centering
    \includegraphics[trim={0cm 0cm 0cm 0.75cm}, clip,width=0.4\textwidth]{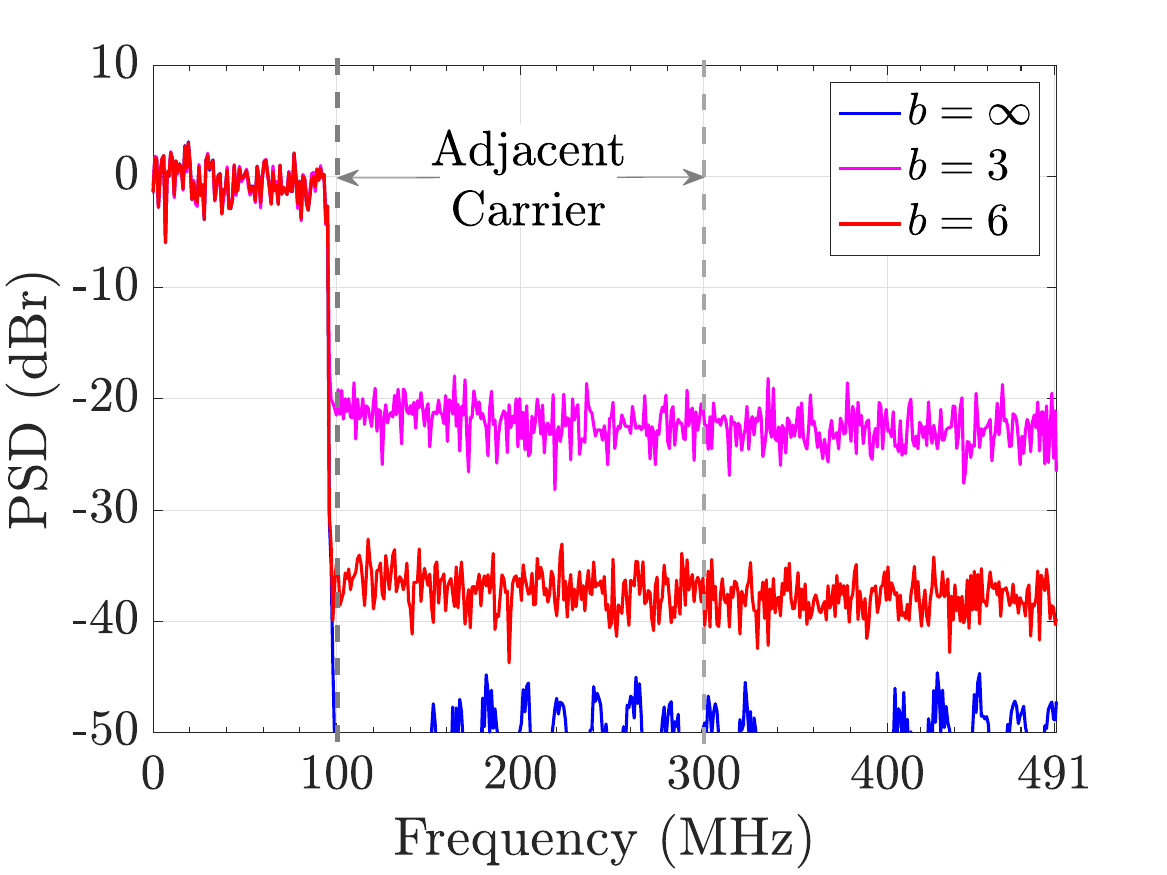}
    \caption{PSD of the linear modulator used for transmitting
    a 400~MHz channel centered at $28$~GHz in a 5G NR system sample rate $f_{\rm samp}= 983$ Ms/s.  The PSD is shown for various number of bits $(n)$ in the DAC.}
    \label{fig:dac_psd}
\end{figure}
}

\begin{figure}
    \centering
    \includegraphics[width=0.38\textwidth]{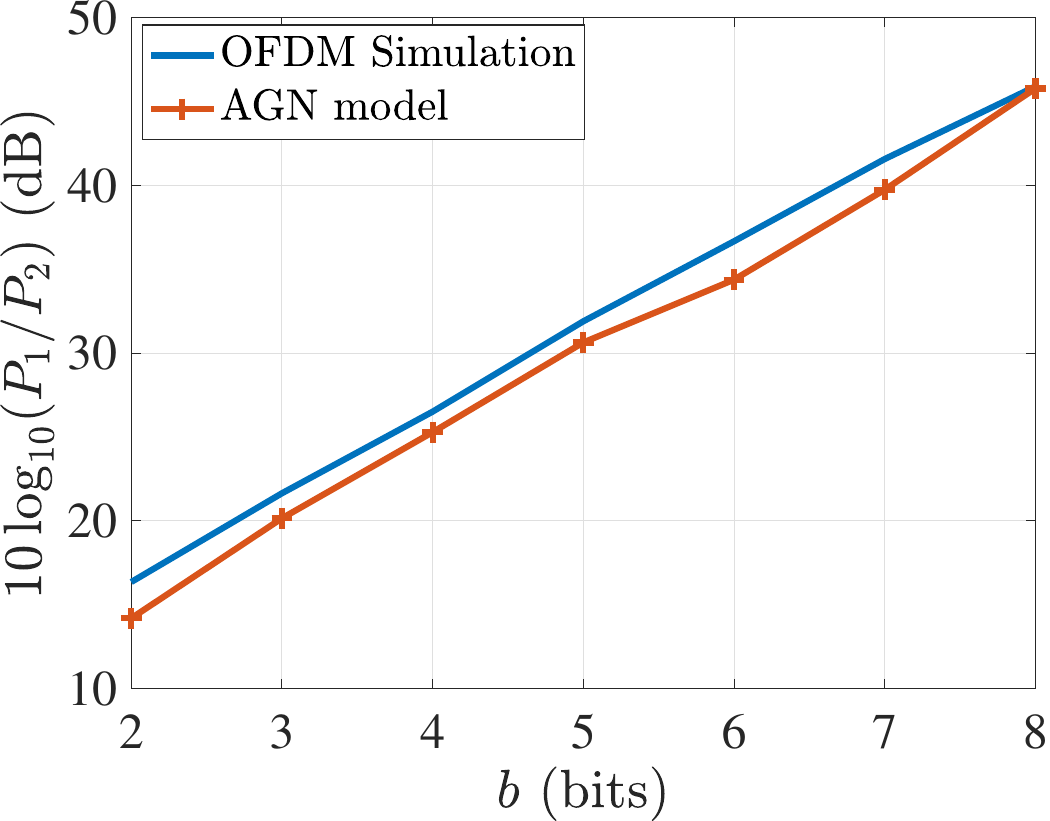}
    \caption{ACLR with a finite DAC resolution ($b$) for a $200$~MHz 3GPP NR OFDM transmitter compared with the proposed AGN model.}
    \label{fig:aclrVsRes}
\end{figure}


\section{Conclusions and Future Work}
We have presented a simple large random limit model
for analyzing the effect of quantization
on a class of linear transceivers.
Importantly, the analysis rigorously captures both
the effects on rate and power spectrum,
including OOB emissions -- key properties
for emerging mmWave systems.  
The analysis confirms earlier simulations that,
for 5G systems, 
low-resolution transceivers cause negligible 
loss in achievable communication rates.  However, OOB emissions
are more problematic. From an information theoretic perspective, this motivates consideration of more advanced modulation
and demodulation methods used in conjunction with low resolution DAC and ADC.  
\iftoggle{conference}{ One approach is to consider
approximate message passing
(AMP) algorithms designed for systems with 
random unitary transforms
\cite{ma2017orthogonal,rangan2019vector,cakmak2014samp,fletcher2018inference,schniter2016vector,he2017generalized,mo2017channel}
and related theoretical results \cite{reeves2017additivity,barbier2019optimal}.
}{ An obvious 
class of methods would be approximate message passing
(AMP) algorithms designed for systems with 
random unitary transforms
\cite{ma2017orthogonal,rangan2019vector,cakmak2014samp,fletcher2018inference,schniter2016vector,he2017generalized}.
These methods indeed
have already been used in mmWave low-resolution receivers \cite{mo2017channel}.
In addition, improved bounds similar to
Theorem~\ref{thm:rate_upper} can likely be
derived from related statistical physical 
techniques that analyze systems exactly of this
form \cite{reeves2017additivity,barbier2019optimal}.
}

\iftoggle{conference}{
\paragraph*{Acknowledgements}
The work  supported in part by
NSF grants  1302336,  1564142,  1547332, and 1824434,  NIST, SRC, and the industrial affiliates of NYU WIRELESS.
}{}

\balance
\bibliographystyle{IEEEtran}
\bibliography{bibl}

\iftoggle{conference}{}{

\newpage

\clearpage
\appendices  

\section{Empirical Convergence of Random Variables}

For the results in Section~\ref{sec:lsl}, we need to first review some technical definitions
on empirical convergence of random variables.  The analysis framework was developed
by Bayati and Montanari \cite{bayati2011dynamics} and also used in the VAMP analysis of \cite{rangan2019vector}.
For a given $p\geq 1$, a map $\gbf:\C^{d}\rightarrow \C^{r}$ is called \emph{pseudo-Lipschitz} of order $p$ if
\beq \label{eq:plfun}
\|\gbf(\rbf_1)-\gbf(\rbf_2)\|\leq C\|\rbf_1-\rbf_2\|\left(1+\|\rbf_1\|^{p-1}+\|\rbf_2\|^{p-1}\right),
\eeq
for some constant $C > 0$.  Note that when $p=1$, we obtain the standard definition of Lipschitz
continuity.

Now suppose that for each $N$, $\xbf(N)$ is a block vector 
$\xbf(N) = (\xbf_1,\ldots,\xbf_N)$ with components $\xbf_n \in \C^d$ for some fixed dimension $d$.
Thus, the total length of the vector is $Nd$.  
Let $X \in \C^d$ be a random vector.  We say that the components $\xbf(N)$ \emph{converge empirically to $X$ with $p$-th
order moments} if
\beq \label{eq:plconv}
    \lim_{N \rightarrow \infty} \frac{1}{N} \sum_{n=0}^{N-1} \phi(\xbf_n(N)) = \Exp\left[ \phi(X) \right],
\eeq
for all pseduo-Lipschitz functions of order $p$.  Loosely speaking, the condition requires that
the empirical distribution of the components of $\xbf(N)$ converge in distribution to the random variable $X$.
The condition will be satisfied when $\xbf_n$ are i.i.d. with distribution $X$.
We will often drop the index $N$ and write,
\beq \label{eq:plp}
    \lim_{N \rightarrow \infty} \{ \xbf_n \} \stackrel{PL(p)}{=} X.
\eeq

\section{Distributions under Random Transforms}

We next need a key result from \cite{rangan2019vector} that describes the distribution of
vectors under random unitary transforms.
Consider a sequence of systems indexed by $N$, and for each $N$ suppose that 
$\Vbf \in \C^{N \times N}$ is uniformly distributed on the unitary matrices.
Let $(\xbf,\sbf)=(\xbf(N),\sbf(N))$ be a sequence of vectors that
converge empirically to random variables $(X,S)$ in that
\beq \label{eq:xs_gen}
    \lim_{N \rightarrow \infty} \{ (x_n,s_n) \} \stackrel{PL(2)}{=} (X,S).
\eeq
Now consider a vector $\ybf$ generated by,
\beq \label{eq:yG_gen}
    \ybf = \Vbf \phibf( \Vbf\herm \xbf,\xibf),
\eeq
where $\phibf(\cdot)$ is some function that operates componentwise in that
\[
    \ybf = \phibf(\zbf,\xibf) \Longleftrightarrow y_n = \phi(z_n,\xi_n),
\]
for some scalar-valued, Lipschitz-continuous function $\phi(\cdot)$.  
Assume that $\xibf$ also converges empirically in that
\[
    \lim_{N \rightarrow \infty} \{ \xi_n \} \stackrel{PL(2)}{=} \Xi,
\]
for some random variable $\Xi$.
To analyze the statitistics on $\ybf$, we define three key quantities:
\begin{subequations} \label{eq:tau_alpha_gen}
\begin{align}
    P &:= \Exp|X|^2, \\
    \alpha &:= \frac{1}{P}\Exp( \overline{Z}\phi(Z,\Xi) ), \\
    \tau_w &:=  \frac{1}{P}\Exp|\phi(Z,\Xi) - \alpha Z|^2.
\end{align}
\end{subequations}
where $Z \sim {\mathcal CN}(0,P)$.

\begin{proposition}  \label{prop:linear}
Under the above assumptions, 
the components of $(\ybf,\xbf,\sbf)$ converge empirically as,
\beq \label{eq:lin_model_lim}
    \lim_{N \rightarrow \infty} \{(y_n,x_n,s_n) \} \stackrel{PL(2)}{=} (Y,X,S),
\eeq
where $(X,S)$ are the random variables in \eqref{eq:xs_gen} and 
\beq \label{eq:lin_model_gen}
    Y = \alpha X + W, \quad W \sim {\mathcal CN}(0,\tau_wP),
\eeq
with $W$ independent of $(X,S)$.
\end{proposition} 
\begin{proof}
This is a special case of one iteration of the general convergence result in \cite[Appendix D]{rangan2019vector}.
That work considers the real-valued case, but the complex case can be proven similarly.
\end{proof}

The model \eqref{eq:lin_model_gen} shows that transformation on $\xbf$ to produce $\ybf$ 
recovers a linearly scaled $\xbf$ plus Gaussian noise. The scaling factor $\alpha$ and 
Gaussian noise variance $\tau_w$ can be computed from the distributions of the components.

\section{Proof of Theorem \ref{thm:spec_fd}} \label{sec:spec_proof}
The theorem is a direct application of the linear model in Propposition~\ref{prop:linear}.
To use the proposition, first observe that,
due to \eqref{eq:delm_limit} and the Gaussian distribution on $\zbf$ in \eqref{eq:zmix},
we have that the sub-band selection $\abf$ and the frequency-domain inputs $\zbf$ converge
empirically as,
\beq \label{eq:alim}
    \lim_{n \rightarrow \infty} \{(z_n,a_n)\} \stackrel{PL(2)}{=} (Z,A),
\eeq
where $A \in \{1,\ldots,M\}$ 
is a discrete random variable with $\Pr(A=m)=\delta_m$ and $Z$ is 
the conditional complex Gaussian, 
\[
    Z \sim C{\mathcal N}(0,P_m) \mbox{ when } A=m.
\]
In particular, the average energy of $Z$ is,
\beq \label{eq:expz_sm}
    \Exp|Z|^2 = \sum_{m=1}^M \delta_m P_m =: \overline{P}.
\eeq
Now, the frequency domain components of the transmitted vector $\xbf$ are given by,
\[
    \rbf = \Vbf\xbf = \Vbf\Qbf(\Vbf\herm \zbf).
\]
Proposition~\ref{prop:linear} then shows that the components of $(\rbf,\zbf,\abf)$ converge empirically as,
\[
    \lim_{N \rightarrow \infty} \{(r_n,z_n,a_n) \} \stackrel{PL(2)}{=} (R,Z,A), 
\]
and
\[
    R = \alpha_{\rm tx}Z + W_{\rm tx}, \quad W_{\rm tx} \sim C{\mathcal N}(0,\tau_{\rm tx}\overline{P}),
\]
where $W_{\rm tx}$ is independent of $Z$.  The sub-band energies,
\begin{align}
    s_m &:= \lim_{N \rightarrow \infty} \sum_{k=0}^{N-1} |r_k|^2\indic{a_k = m} \nonumber \\
        &= \Exp\left[ |R|^2 \indic{A = m} \right] \nonumber \\
    &= \Exp\left[ |\alpha_{\rm tx}Z + W_{\rm tx}|^2|A=m\right]\Pr(A=m) \nonumber \\
    &= \left[ |\alpha_{\rm tx}|^2P_m + \tau_{tx}\overline{P} \right] \delta_m.
\end{align}
This proves \eqref{eq:sm_fd}.
To prove \eqref{eq:stot_fd},
\begin{align*}
    s_{\rm tot} &= \sum_{m=1}^M s_m = \sum_{m=1}^M \left[ |\alpha_{\rm tx}|^2P_m + \tau_{tx}\overline{P} \right] \delta_m \nonumber \\
    &= (|\alpha_{\rm tx}|^2+ \tau_{tx})\overline{P},
\end{align*}
where the last step used \eqref{eq:expz_sm} and the fact that $\sum_m \delta_m = 1$.

\section{The Linear Rate Region} \label{sec:linear_feas}
The following proposition shows that power allocations
$\nu_m$ are feasible if and only if they satisfy 
\eqref{eq:nu_fd_feas}.

\begin{proposition} \label{prop:nu_feas}
Let $\alpha_{\rm tx} \in \C$, $\tau_{\rm tx} > 0$,
$\overline{P} > 0$ and
$\delta_m \geq 0$ with $\sum_m \delta_m = 1$ be given.
For any $\nubf=(\nu_1,\ldots,\nu_M)$, 
the following are equivalent:
\begin{enumerate}[(a)]
\item There exists $P_m \geq 0$ such that 
$\overline{P} = \sum_m \delta_m P_m$ and $\nu_m$ is given by \eqref{eq:nu_fd} for all $m$.
\item $\nu_m$ satisfies \eqref{eq:nu_fd_feas} for all $m$
and $\sum_m \nu_m = 1$.
\end{enumerate}
\end{proposition}
\begin{proof}
$(a) \Rightarrow (b)$:
Suppose there exists $P_m \geq 0$ as in (a) and let $\nu_m$
be given by \eqref{eq:nu_fd}.
Since $\overline{P} = \sum_m \delta_m P_m$, we have $\sum_m \nu_m =1$.  Also, since $P_m \geq 0$, we have $\nu_m$ in \eqref{eq:nu_fd} satisfies the lower bound \eqref{eq:nu_fd_feas}.

$(b) \Rightarrow (a)$: Conversely, suppose we are given $\nubf$ satisfying
 \eqref{eq:nu_fd_feas} with $\sum_m \nu_m = 1$.
Set,
\beq \label{eq:Pm_nu}
    P_m = \left[\frac{\nu_m}{\delta_m}(|\alpha_{\rm tx}|^2 + \tau_{\rm tx}) - \tau_{\rm tx}\right] 
    \frac{\overline{P}}{|\alpha_{\rm tx}|^2}.
\eeq
Therefore, $\nu_m$ satisfies \eqref{eq:nu_fd}.
Since $\nu_m$ satisfies \eqref{eq:nu_fd_feas}, 
$P_m$ in \eqref{eq:Pm_nu} satisfies $P_m \geq 0$. Also,
\begin{align*}
    \MoveEqLeft \sum_m \delta_m P_m =
    \frac{\overline{P}}{|\alpha_{\rm tx}|^2} \left[ 
    \sum_m \nu_m (|\alpha_{\rm tx}|^2 + \tau_{\rm tx})
    - \sum_m \delta_m \tau_{\rm tx} \right] \\
    &= \frac{\overline{P}}{|\alpha_{\rm tx}|^2} \left[ 
    |\alpha_{\rm tx}|^2 + \tau_{\rm tx}
    - \tau_{\rm tx} \right] = \overline{P},
\end{align*}
where we have used the fact that $\sum_m \nu_m = 1$
and $\sum_m \delta_m = 1$.
\end{proof}

\section{Proof of Theorem \ref{thm:rate_lin}} \label{sec:rate_lin_proof}

We need two basic mutual information lemmas.
For $m=1,\ldots,M$, let $\zbf^{(m)}$ and $\zbfhat^{(m)}$ denote the sub-vectors of $\zbf$ and $\zbfhat$
with components in sub-band $m$.  That is, 
\[
    \zbf^{(m)} = \{ z_k | a_k = m \},
\]
and $\zbfhat^{(m)}$ is defined similarly.

\begin{lemma} \label{lem:chain}  The mututal information is bounded below by,
\beq \label{eq:mi_sum}
    I(\zbf;\zbfhat) \geq \sum_{m=1}^M I(\zbf^{(m)};\zbfhat^{(m)}).
\eeq
\end{lemma}
\begin{proof}
By the mutual information chain rule,
\beq \label{eq:mi_sum1}
    I(\zbf;\zbfhat) \geq \sum_{m=1}^M I(\zbf^{(m)};\zbfhat|\zbf^{(1)},\ldots,\zbf^{(m-1)}).
\eeq
Also, since the components $\zbf$ are independent, the vectors $\zbf^{(m)}$ are 
independent for different $m$.  Hence,
\beq \label{eq:hm}
    H(\zbf^{(m)}|\zbf^{(1)},\ldots,\zbf^{(m-1)}) 
    = H(\zbf^{(m)}).
\eeq
Therefore,
\begin{align} 
    \MoveEqLeft I(\zbf^{(m)};\zbfhat|\zbf^{(1)},\ldots,\zbf^{(m-1)}) = 
        H(\zbf^{(m)}|\zbf^{(1)},\ldots,\zbf^{(m-1)}) \nonumber \\
        &-
        H(\zbf^{(m)}|\zbfhat,\zbf^{(1)},\ldots,\zbf^{(m-1)}) \nonumber \\
        &\stackrel{(a)}{=}
        H(\zbf^{(m)}) - H(\zbf^{(m)}|\zbfhat,\zbf^{(1)},\ldots,\zbf^{(m-1)}) \nonumber \\
        &\stackrel{(b)}{\geq}
        H(\zbf^{(m)}) - H(\zbf^{(m)}|\zbfhat) \nonumber \\
        &\stackrel{(c)}{\geq}
        H(\zbf^{(m)}) - H(\zbf^{(m)}|\zbfhat^{(m)}) \nonumber \\
        &= I(\zbf^{(m)}|\zbfhat^{(m)}),
        \label{eq:mi_bnd1}
\end{align}
where (a) follows from \eqref{eq:hm}, and (b) and (c) follows from the fact
that conditioning always reduce the entropy.  Substituting \eqref{eq:mi_bnd1} 
into \eqref{eq:hm} proves \eqref{eq:mi_sum}.
\end{proof}

\begin{lemma} \label{lem:mi_lower}  Suppose that $\zbf \in \C^d$ is a complex Gaussian random vector 
with i.i.d.\ components $z_k \sim C{\mathcal N}(0,P)$.  Let $\ybf$ be any other random vector 
with correlation coefficient,
\[
    \rho := \frac{|\Exp(\zbf\herm \ybf)|^2}{\Exp\|\ybf\|^2\Exp\|\zbf\|^2} = \frac{|\Exp(\zbf\herm \ybf)|^2}{\Exp\|\ybf\|^2Pd}.
\]
Then, the mutual information between $\zbf$ and $\ybf$ is bounded below by,
\[
    I(\zbf;\ybf) \geq  -d \ln(1-\rho).
\]
\end{lemma}
\begin{proof}
The mutual information is,
\beq \label{eq:Igauss}
    I(\zbf;\ybf) = H(\zbf) - H(\zbf|\ybf).
\eeq
Since $\zbf$ is i.i.d.\ with $d$ components distributed as $C{\mathcal N}(0,P)$, 
\beq \label{eq:Hgauss}
    H(\zbf) = d\ln(\pi e P).
\eeq
Now, given $\ybf$, $\zbf$ will have a conditional variance,
\beq
    \sigma^2 := \frac{1}{d}\Exp\left[ \|\zbf - \zbfhat(\ybf)\|^2\right],
\eeq
where $\zbfhat(\ybf)$ is the MMSE estimator of $\zbf$ given $\ybf$.
So, the conditional entropy $H(\zbf|\ybf)$ is bounded below by
the entropy of the Gaussian,
\beq \label{eq:Hzy1}
    H(\zbf|\ybf) \leq d\ln(\pi e \sigma^2).
\eeq
But, we can further bound $H(\zbf|\ybf)$ by replacing $\sigma^2$ with the variance
for a linear estimator,
\beq \label{eq:sig_bnd}
    \sigma^2 \leq \frac{1}{d} \left[ \Exp\|\zbf\|^2 - \frac{\Exp\|\zbf\herm \ybf\|^2}{\Exp\|\ybf\|^2} \right] = 
    \frac{1}{d} \left[ P - \frac{|\Exp(\zbf\herm \ybf)|^2}{\Exp\|\ybf\|^2} \right].
\eeq
Therefore, substituting \eqref{eq:Hgauss}, \eqref{eq:Hzy1} and \eqref{eq:sig_bnd} into \eqref{eq:Igauss},
\begin{align}
    I(\zbf;\ybf) \geq d\ln(\pi e P) - d \ln(\pi e \sigma^2) \geq -d\ln(1 - \rho). \nonumber
\end{align}
\end{proof}

We use these lemmas as follows.  In each sub-band $m$, the components of $\zbf^{(m)}$
are i.i.d.\ complex Gaussians with zero mean and variance $P_m$.
So, by Lemma~\ref{lem:mi_lower}, 
\beq \label{eq:Im1}
    I(\zbf^{(m)}; \zbfhat^{(m)}) \geq -N_m \ln(1 - \rho^{(m)}),
\eeq
where $N_m$ is the number of coefficients in sub-band $m$ and
$\rho^{(m)}$ is the correlation coefficient,
\beq \label{eq:rhom}
    \rho^{(m)} := \frac{|\Exp\left( (\zbf^{(m)})\herm \zbfhat^{(m)}\right)|^2}{\Exp\|\zbfhat^{(m)}\|^2P_m}.
\eeq
Now, \eqref{eq:delm_limit} shows that $N_m/N \rightarrow \delta_m$.
So, if we divide \eqref{eq:Im1} by $N$ and take the limit we get,
\beq \label{eq:Im2}
    \liminf_{N \rightarrow \infty} \frac{1}{N} I(\zbf^{(m)}; \zbfhat^{(m)}) \geq -\delta_m  \ln(1 - \overline{\rho}^{(m)}).
\eeq
where $\overline{\rho}^{(m)}$ is the limiting correlation,
\beq \label{eq:rho_lim}
    \overline{\rho}^{(m)} := \lim_{N \rightarrow \infty} \rho^{(m)}
\eeq

To compute the limiting correlation in \eqref{eq:rho_lim}, we use 
a similar calculation to the proof of Theorem \ref{thm:spec_fd}.
Specifically, the received symbols are given by,
\[
    \zbfhat = \Vbf\Gbf(\Vbf\herm \zbf + \xibf).
\]
Proposition~\ref{prop:linear} then shows that the components of $(\zbfhat,\zbf,\abf)$ converge empirically as,
\[
    \lim_{N \rightarrow \infty} \{(\widehat{z}_n,z_n,a_n) \} \stackrel{PL(2)}{=} (\widehat{Z},Z,A), 
\]
and
\[
    \widehat{Z} = \alpha_{\rm rx}Z + W_{\rm rx}, \quad W_{\rm rx} \sim C{\mathcal N}(0,\tau_{\rm rx}\overline{P}),
\]
where $W_{\rm rx}$ is independent of $Z$.  Now, we have that,
\begin{align}
   \MoveEqLeft \lim_{N \rightarrow \infty} \frac{1}{N} (\zbf^{(m)})\herm \zbfhat^{(m)}  = 
    \lim_{N \rightarrow \infty} \frac{1}{N} z^*_k \zhat_k\indic{a_k = m}  \nonumber \\
    &= \Exp\left[ Z^*\Zhat \indic{A_k = m} \right] =  \Exp\left[ \overline{Z}\Zhat|A=m \right]\Pr(A=m) \nonumber \\
    &= \Exp\left[ Z^*(\alpha_{\rm rx}Z + W_{\rm rx}) |A=m\right] \delta_m \nonumber \\
    &= \alpha_{\rm rx}P_m \delta_m, \nonumber 
\end{align}
where we have used that, conditional on $A=m$, $\Exp|Z|^2=P_m$ and $\Exp(Z^*W_{\rm rx})=0$.  
Hence,
\beq \label{eq:zzhat_lim1}
    \lim_{N \rightarrow \infty} \frac{1}{N^2} \left| \Exp (\zbf^{(m)})\herm \zbfhat^{(m)} \right|^2 = |\alpha_{\rm rx}|^2P_m^2 \delta_m^2.
\eeq
Similar calculations show that,
\begin{align}
    \lim_{N \rightarrow \infty} \frac{1}{N}  \Exp \|\zbfhat^{(m)} \|^2 &= \left[|\alpha_{\rm rx}|^2P_m 
        + \tau_{\rm rx}\overline{P}\right]\delta_m 
        \label{eq:zzhat_lim2} \\
    \lim_{N \rightarrow \infty} \frac{1}{N}  \Exp \|\zbf^{(m)} \|^2 &= P_m\delta_m. \label{eq:zzhat_lim3}
\end{align}
Substituting \eqref{eq:zzhat_lim1}, \eqref{eq:zzhat_lim2} and \eqref{eq:zzhat_lim3} into \eqref{eq:rhom}, we obtain that the limit
in \eqref{eq:rho_lim} is given by,
\beq \label{eq:rho_lim}
    \overline{\rho}^{(m)} := \frac{|\alpha_{\rm rx}|^2P_m^2 \delta_m^2}{P_m\left[|\alpha_{\rm rx}|^2P_m 
        + \tau_{\rm rx}\overline{P}\right]\delta_m^2}
    = \frac{|\alpha_{\rm rx}|^2P_m}{|\alpha_{\rm rx}|^2P_m + \tau_{\rm rx}\overline{P}}.
\eeq
Hence, from \eqref{eq:Im2}, we obtain
\beq \label{eq:Im3}
    \liminf_{N \rightarrow \infty} \frac{1}{N} I(\zbf^{(m)}; \zbfhat^{(m)}) \geq \delta_m  \log\left(1 + \frac{|\alpha_{\rm rx}|^2P_m}{\tau_{\rm rx}\overline{P}}\right).
\eeq
Substituting \eqref{eq:Im3} into the sum \eqref{eq:mi_sum} obtains \eqref{eq:rate_lin}.

\section{Proof of Theorem~\ref{thm:rate_no_noise}}
\label{sec:rate_no_noise_proof}
This is a straightforward mathematical manipulation
\begin{align*} \label{eq:rate_lin_no_noise1}
    R_{\rm lin} &\stackrel{(a)}{\geq}
    \sum_{m=1}^M \delta_m \log\left(1 + \frac{|\alpha_{\rm tx}|^2P_m}{\tau_{\rm tx} \overline{P}} \right) \nonumber \\
    &=
    \sum_{m=1}^M \delta_m \log\left(\frac{\tau_{tx}\overline{P} + |\alpha_{\rm tx}|^2P_m}{\tau_{\rm tx} \overline{P}} \right)
        \nonumber \\
    &\stackrel{(b)}{=} \sum_{m=1}^M \delta_m \log\left(\frac{\nu_m(|\alpha_{\rm tx}|^2 + \tau_{\rm tx})\overline{P}}{\delta_m \tau_{\rm tx} \overline{P}} \right) \nonumber \\
    &= \log\left( 1 + \frac{|\alpha_{\rm tx}|^2}{\tau_{\rm tx}} \right) + \sum_{m=1}^M \delta_m \log\left(\frac{\nu_m}{\delta_m}\right)
        \nonumber \\
    &= \log\left( 1 + \frac{|\alpha_{\rm tx}|^2}{\tau_{\rm tx}} \right) - D(\deltabf \| \nubf).
\end{align*}
where (a) follows from \eqref{eq:rate_lin_awgn}
with $\sigma^2=0$ and (b)
follows from \eqref{eq:nu_fd}.
This proves     \eqref{eq:rate_no_noise}.

\section{Proof of Theorem~\ref{thm:rate_upper}}
\label{sec:proof_rate}

We first need a lemma to characterize the maximum entropy, $H_{\rm max}(s)$
in \eqref{eq:Hmax_def}.
Let $S$ be a random variable given by $|x|^2$ where $x$ uniformly
distributed on the set of DAC constellation points $x\in A$.
Hence, $\Exp(S)$ is the average energy per sample if the modulator
uniformly selects sequences from the DAC output.  
Let $\lambda_S(\theta)$ be its cummulant generating function,
\beq \label{eq:lamS}
    \lambda_S(\theta) = \log \Exp \exp(\theta S) = 
    \log\left( \frac{1}{|A|}\sum_{x\in A}
        e^{\theta|x|^2} \right),
\eeq
and, let $I_S(s)$ be its Legendre transform,
\beq \label{eq:IS}
    I_S(s) := \sup_{\theta} \left[ \theta s - \lambda(\theta) \right].
\eeq

\begin{lemma}  \label{lem:Hmax} The maximum entropy in \eqref{eq:Hmax_def} is given by,
\beq \label{eq:Hmax}
    H_{\rm max}(s) = \log|A| - I_S(s).
\eeq
\end{lemma}
\begin{proof}
Consider a set of distributions of a discrete random variable $V_\theta$
on the set $A$, parameterized by the scalar real variable $\theta$, where
\beq \label{eq:Vthetapmf}
    P(V_\theta=x) = \frac{1}{|A|Z(\theta)} e^{\theta |x|^2},
    \quad
    Z(\theta) = \frac{1}{|A|}\sum_{x \in A} e^{\theta |x|^2}.
\eeq
It is known that if $\theta$ is selected such that
$\Exp|V_\theta|^2 = s$, then $V_\theta$ is the maximum entropy distribution
over all random variables $V$ on $A$ with $\Exp|V|^2 = s$.
Observe that the cummulant generating function \eqref{eq:lamS} is,
\[
    \lambda_S(\theta) = \log Z(\theta).
\]
A standard result on exponential families is that,
\[
    \lambda_S'(\theta) = \Exp|V_\theta|^2.
\]
Now, for any $s$, we have
\[
    I_S(s) = \widehat{\theta} s - \lambda_S(\widehat{\theta}),
\]
where
\beq \label{eq:theta_max}
    \widehat{\theta} = \argmax_{\theta} \left[ s\theta - \lambda_S(\theta) \right].
\eeq
Since $\widehat{\theta}$ is the maximizer in \eqref{eq:theta_max},
\[
    \lambda_S'(\widehat{\theta}) = s \Rightarrow s = \Exp|V_{\widehat{\theta}}|^2.
\]
So, $V_{\widehat{\theta}}$ is the maximum entropy distribution with $E|V|^2=s$.
Also, the entropy of PMF of $V_{\widehat{\theta}}$ in \eqref{eq:Vthetapmf} is,
\begin{align}
    H_{\rm max}(s) &= H(V_{\widehat{\theta}}) = -\Exp \log P(V_{\widehat{\theta}}) \nonumber \\ 
    &= \log|A| + \log Z(\widehat{\theta}) - \widehat{\theta}\Exp |V_{\widehat{\theta}}|^2
    \nonumber \\
    & = \log|A| + \log \lambda_S(\widehat{\theta}) - \widehat{\theta}s \nonumber \\
    & = \log|A| - I_S(s).
\end{align}
\end{proof}

We now proceed to the proof of Theorem~\ref{thm:rate_upper}.
There are $|A|^N$ sequences in the set $A^N$.  So,
if we let $\xbf_N$ be the random vector uniformly generated on $A^N$,
we have that the expected cardinality of the set $G_N(\Vbf,\epsilon)$ in
\eqref{eq:Fneps} is,
\[
    \Exp|G_N(\Vbf,\epsilon)| = |A|^N \Pr(\phi(\Vbf\xbf_N) \in [\sbf-\epsilon,\sbf]),
\]
where the probability is taken over the random vector $\xbf_N$
and the matrix $\Vbf$.
Hence, the rate upper bound in \eqref{eq:rate_upper} is 
\beq \label{eq:rate_upper1}
    \overline{R} = \ln|A| + 
    \lim_{\epsilon \rightarrow 0} \lim_{N \rightarrow \infty} \frac{1}{N} \ln \Pr( \phi(\Vbf\xbf_N) \in [\sbf-\epsilon,\sbf]).
\eeq
So, we need to compute a tail probability.  This is a standard large deviations calculation.
Define the random variable,
\beq \label{eq:Sndef}
    S_N := \frac{1}{N}\|\xbf_N\|^2, 
\eeq
which represents the per sample total energy in the vector
$\xbf_N$.  Also, let $\ubf_N$ be the unit vector,
\beq \label{eq:undef}
    \ubf_N = \frac{1}{\|\xbf_N\|}\Vbf\xbf_N = 
    \frac{1}{\|\Vbf\xbf_N\|}\Vbf\xbf_N.
\eeq
Since $\Vbf$ is Haar distributed on the unitaries, $\ubf_N$ is
uniformly distributed on the sphere of radius one and independent
of $\xbf_N$.  Also, let
\beq \label{eq:nu_Nm}
    \nu_{N,m} := \frac{1}{N} \sum_{k=0}^{N-1} \indic{a_k=m} |u_k|^2,
\eeq
which is the fraction of the energy of $\ubf$ in sub-band $m$.
With these definitions, if $\rbf=\Vbf\xbf$
the sub-band energies in \eqref{eq:sm_nolim} is
given by,
\begin{align}
    \phi_m(\rbf) &:= \frac{1}{N}\sum_{k=0}^{N-1} |r_k|^2\indic{a_k=m} \nonumber \\
        &\stackrel{(a)}{=} \frac{1}{N}\|\xbf\|^2 \sum_{k=0}^{N-1} |u_k|^2\indic{a_k=m}  
        \stackrel{(b)}{=} S_N \nu_{N,m}
\end{align}
where (a) holds since $\rbf = \Vbf\xbf = \|\xbf\|^2\ubf$,
and (b) holds from the definition of $S_N$ in $\nu_{N,m}$ in \eqref{eq:Sndef} and \eqref{eq:nu_Nm}.
So, $\xbf_N \in G_N(\Vbf,\epsilon)$ if and only if,
\beq \label{eq:Snu_con}
     S_N \nu_{N,m}  \in [s_m-\epsilon, s_m]
\eeq
for all $m$.  
Therefore, if we define the set, 
\beq \label{eq:Gdef}
    G_\epsilon := \left\{ (s,\nubf) ~|~ s\nu_m \in  [s_m-\epsilon, s_m] \right\},
\eeq
the constraint \eqref{eq:Snu_con} can be written as $(S_N,\nubf_N) \in G$,
and the rate upper bound \eqref{eq:rate_upper1} is given by,
\beq \label{eq:rate_upper2}
    \overline{R} = \ln|A| + 
    \lim_{\epsilon \rightarrow 0} \lim_{N \rightarrow \infty} \frac{1}{N} \ln \Pr\left( (S_N,\nubf_N) \in G_\epsilon \right).
\eeq
We will calculate the probability  using large deviations.

First, since $S_N$ in \eqref{eq:Sndef} is given by,
\[
    S_N = \frac{1}{N} \sum_{n=0}^{N-1} |x_n|^2,
\]
which is an empirical average of i.i.d.\ random variables with distribution $S$, the random variable $|x|^2$
where $x$ is uniformly distributed on the DAC constellation points $A$.
By Cramer's theorem, it satisfies the large deviations principle (LDP) 
with rate function, $I_S(s)$ \cite{dembo2010large_deviations}.

Also, each component $\nu_{N,m}$ is the energy fraction of the
projection of an $N$-dimensional complex unit vector 
onto a sub-space of dimension $N_m$ with $N_m = \delta_m N$.  
Thus, $\nubf_N$ has the Dirchelet distribution with probability density,
\[
    p(\nubf_N) = \frac{1}{B(\alphabf)}\prod_{m=1}^M 
    \nu_{N,m}^{\alpha_m - 1},
\]
where $\alphabf$ is the vector with coefficients,
\[
    \alpha_m = N\delta_m,
\]
and
\[
    B(\alphabf) = \frac{\prod_m \Gamma(\alpha_m)}{\Gamma(\sum_m \alpha_m)}.
\]
Using Sterling's approximation for large $N$, the density
is approximately given by,
\[
     p(\nubf_N) \approx \exp\left[ -N D(\deltabf\|\nubf) \right].
\]
Therefore, $\nubf_N$ satisfies the LDP with rate function
\[
    I_{\nubf}(\nubf) = D(\deltabf\|\nubf).
\]

Since $S_N$ and $\nubf_N$ are independent, they have a rate $I_S(s) + I_{\nubf}(\nubf)$.
By the property of the rate function,
\begin{align}
    \MoveEqLeft \lim_{N \rightarrow \infty} \frac{1}{N} \ln \Pr\left( (S_N,\nubf_N) \in G_\epsilon \right) 
    \nonumber \\
    &= -\inf_{s,\nubf \in G_\epsilon} \left[  I_S(s) + D(\deltabf\|\nubf) \right].
\end{align}
Using the definition of \eqref{eq:Gdef} and the fact that $I_S(s)$ is continuous,
we obtain
\begin{align}
    \MoveEqLeft \lim_{\epsilon \rightarrow 0} \lim_{N \rightarrow \infty} 
        \frac{1}{N} \ln \Pr\left( (S_N,\nubf_N) \in G_\epsilon \right) 
    \nonumber \\
    &= -\inf_{s,\nubf \in G_0} \left[  I_S(s) + D(\deltabf\|\nubf) \right].
\end{align}
But, taking $\epsilon = 0$ in \eqref{eq:Gdef}, we see that $G_0$ is the set
\[
     G_\epsilon := \left\{ (S,\nubf) ~|~ s\nu_m = s_m ~\forall m \right\}.
\]
Since $\sum_m \nu_m = 1$, the only point in $G_0$ are the $s=s_{\rm tot}=\sum_m s_m$
and $\nu_m = s_m/s_{\rm tot}$.  Therefore,
\begin{align*}
   \MoveEqLeft \lim_{\epsilon \rightarrow 0} \lim_{N \rightarrow \infty} 
        \frac{1}{N} \ln \Pr\left( (S_N,\nubf_N) \in G_\epsilon \right) \nonumber \\
    &= -I_S(s_{\rm tot}) - D(\deltabf\|\nubf).
\end{align*}
Substituting this into \eqref{eq:rate_upper2},
\beq \label{eq:rate_upper3}
    \overline{R} = \ln|A| -  I_S(s_{\rm tot}) - D(\deltabf\|\nubf).
\eeq
From Lemma~\ref{lem:Hmax}, this can be re-written as,
\beq \label{eq:rate_upper4}
    \overline{R} = H_{\rm max}(s_{\rm tot}) - D(\deltabf\|\nubf).
\eeq
}

\end{document}